\documentclass{llncs} 

\usepackage{amssymb}
\usepackage{amsmath}
\usepackage{psboxit}
\usepackage{graphics}

\newcommand{\rr}{{\cal R}}
\newcommand{\pp}{{\cal P}}
\newcommand{\gateway}{Gateway }

\begin{document}

\author{Sara Miner More \and Pavel Naumov}

\institute{ Department of Mathematics and Computer Science\\ McDaniel College, Westminster, Maryland 21157, USA
\email{ \{smore,pnaumov\}@mcdaniel.edu}}
\title{Functional Dependence of Secrets in a Collaboration Network}

\maketitle

\begin{abstract}
A collaboration network is a graph formed by communication channels between parties. Parties communicate over these channels
to establish secrets, simultaneously enforcing interdependencies between the secrets. The paper studies properties of these interdependencies that are induced by the topology of the network. In previous work, the authors developed a complete logical system for one such property, independence, also known in the information flow literature as nondeducibility.  This work describes a complete and decidable logical system for the functional dependence relation between sets of secrets over a collaboration network. The system extends Armstrong's system of axioms for functional dependency in databases.
\end{abstract}

\section{Introduction}
In this paper, we study properties of interdependencies between pieces of information. We call these pieces {\em secrets} to emphasize the fact that they might be unknown to some parties.  Below, we first describe two relations for expressing interdependencies between secrets.  Next, we discuss these relations in the context of collaboration networks which specify the available communication channels for the parties establishing the secrets.

\subsection{Relations on Secrets}
One of the simplest relations between two secrets is {\em functional dependence}, which we denote by $a\rhd b$. This means that the value of secret $a$ reveals the value of secret $b$. This relation is reflexive and transitive. 
A more general and less trivial form of functional dependence is functional dependence between sets of secrets. If $A$ and $B$ are two sets of secrets, then $A\rhd B$ means that, together, the values of all secrets in $A$ reveal the values of all secrets in $B$. Armstrong~\cite{a74} presented the following sound and complete axiomatization of this relation:
\begin{enumerate}
\item {\em Reflexivity}: $A\rhd B$, if $A\supseteq B$,
\item {\em Augmentation}: $A\rhd B \rightarrow A,C\rhd B,C$,
\item {\em Transitivity}: $A\rhd B \rightarrow (B\rhd C \rightarrow A\rhd C)$,
\end{enumerate}
where here and everywhere below $A,B$ denotes the union of sets $A$ and $B$. The above axioms are known in database literature as Armstrong's axioms \cite[p.~81]{guw09}. Beeri, Fagin, and Howard~\cite{bfh77} suggested a variation of Armstrong's axioms that describe properties of multi-valued dependency.

Not all dependencies between two secrets are functional. For example, if secret $a$ is a pair $\langle x, y\rangle$ and 
secret $b$ is a pair $\langle y, z\rangle$, then there is an interdependency between these secrets in the sense that not every value of secret $a$ is compatible with every value of secret $b$. However, neither $a\rhd b$ nor $b\rhd a$ is necessarily true. If there is no interdependency between two secrets, then we will say that the two secrets are {\em independent}. In other words, secrets $a$ and $b$ are independent if any possible value of secret $a$ is compatible with any possible value of secret $b$. We denote this relation between two secrets by $a\parallel b$. This relation was introduced by Sutherland~\cite{s86} and is also known as {\em nondeducibility} in the study of information flow. Halpern and O'Neill~\cite{ho08} proposed a closely related notion called $f$-secrecy. 

Like functional dependence, independence also can be generalized to relate two sets of secrets. If $A$ and $B$ are two such sets, then $A\parallel B$ means that any consistent combination of values of the secrets in $A$ is compatible with any consistent combination of values of the secrets in $B$. Note that ``consistent combination" is an important condition here, since some interdependency may exist between secrets in set $A$ even while the entire set of secrets $A$ is independent from the secrets in set $B$. A sound and complete axiomatization of this independence relation between sets was given by More and Naumov~\cite{mn10}:
\begin{enumerate}
\item {\em Empty Set}: $\varnothing\parallel A$,
\item {\em Monotonicity}: $A,B\parallel C\rightarrow A\parallel C$,
\item {\em Symmetry}: $A\parallel B \rightarrow B\parallel A$,
\item {\em Public Knowledge}: $A\parallel A \rightarrow (B\parallel C \rightarrow A,B\parallel C)$,
\item {\em Exchange}: $A,B\parallel C\rightarrow (A\parallel B\rightarrow A \parallel B,C)$.
\end{enumerate}
The assumption $A\parallel A$ in the Public Knowledge axiom guarantees that each secret in the set $A$ has a fixed value and, thus, is ``public knowledge". Details can be found in the original work~\cite{mn10}. Essentially the same axioms were shown by Geiger, Paz, and Pearl~\cite{gpp91} to provide a complete axiomatization of the independence relation between sets of random variables in probability theory. 

A complete logical system that combines the independence and functional dependence predicates for {\em single} secrets was described by Kelvey, More, Naumov, and Sapp~\cite{kmns10}:
\begin{enumerate}
\item {\em Reflexivity:} $a\rhd a$,
\item {\em Transitivity:} $a\rhd b \rightarrow (b\rhd c \rightarrow a \rhd c)$,
\item {\em Symmetry:} $a \parallel b \rightarrow b \parallel a$,
\item {\em Universal Independence:} $a\parallel a \rightarrow a \parallel b$, 
\item {\em Universal Dependence:} $a\parallel a \rightarrow b \rhd a$,
\item {\em Substitution:} $a\parallel b \rightarrow (b\rhd c \rightarrow a\parallel c)$,
\end{enumerate}
where $a,b$ and $c$, unlike $A,B$ and $C$ above, stand for single secrets, not sets of secrets.

\subsection{Secrets in Collaboration Networks}

So far, we have assumed that the values of secrets are determined a priori. In the physical world, however, secret values are often generated, or at least disseminated, via interaction between several parties. Quite often such interaction happens over a fixed network. For example, in social networks, interaction between nodes happens along connections formed by friendship, kinship, financial relationship, etc.  In distributed computer systems, interaction happens over computer networks. Exchange of genetic information happens along the edges of the genealogical tree. Corporate secrets normally flow over an organization chart. In cryptographic protocols, it is often assumed that values are transmitted over well-defined channels. On social networking websites, information is shared between ``friends". Messages between objects on an UML interaction diagram are sent along connections defined by associations between the classes of the objects.

We attempt to capture this type of information flow over a graph by the notion of a {\em collaboration network}. Such a network consists of several parties connected by communication channels that form a network with a fixed topology. A pair of parties connected by a channel uses this channel to establish a secret. If the pairs of parties establish their secrets completely independently from other pairs, then possession of one or several of these secrets reveals no information about the other secrets. Assume, however, that secrets are not picked completely independently. Instead, each party with access to multiple channels may enforce some desired interdependency between the secrets it shares with other parties. These ``local" interdependencies between secrets known to a single party may result in a ``global" interdependency between several secrets, not all of which are known to any single party. Given the fixed topology of the collaboration network, we study what global interdependencies between secrets may exist in the system.
\begin{figure}[htbp]
   \centering
	\scalebox{.5}{\includegraphics{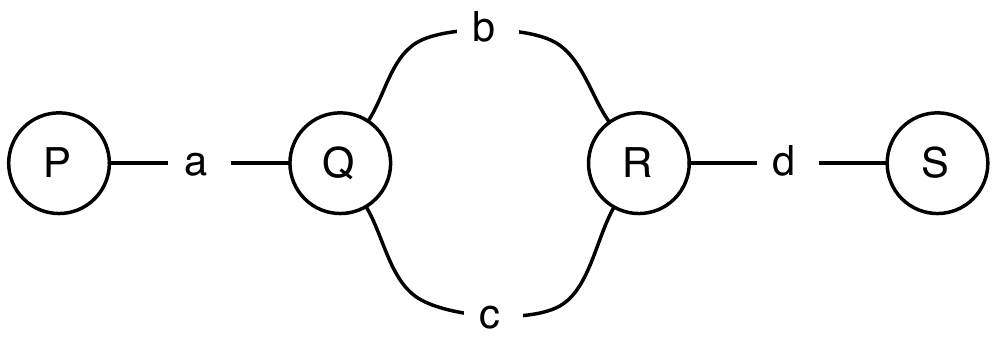}}
   \caption{Collaboration network $N_1$.}
   \label{computation_graph}
\end{figure}

Consider, for example, the collaboration network $N_1$ depicted in Figure~\ref{computation_graph}.  Suppose that the parties collaborate according to the following protocol.  Party $P$ picks a random value $a$ from $\{0,1\}$ and sends it to party $Q$. Party $Q$ picks values $b$ and $c$ from $\{0,1\}$ in such a way that $a=b+c \mod 2$ and sends both of these values to $R$. Party $R$ computes $d=b+c \mod 2$ and sends value $d$ to party $S$. In this protocol, it is clear that the values of $a$ and $d$ will always match. Hence, for this specific protocol, we can say that $a\rhd d$, but at the same time $a\parallel b$ and $a\parallel c$.

Note that in the above example, all channels transmit secret messages in one direction and, thus, the channel network forms a directed graph. However, in the more general setting, two parties might establish the value of a secret through a dialog over their communication channel, with messages traveling in both directions. Thus, in general, we will not assume any specific direction on a channel.


\subsection{Data Streams and Collaboration Networks}
In this section, we will consider a more sophisticated example of collaboration network from network coding theory.

\begin{figure}[htbp]
   \centering
	\scalebox{.5}{\includegraphics{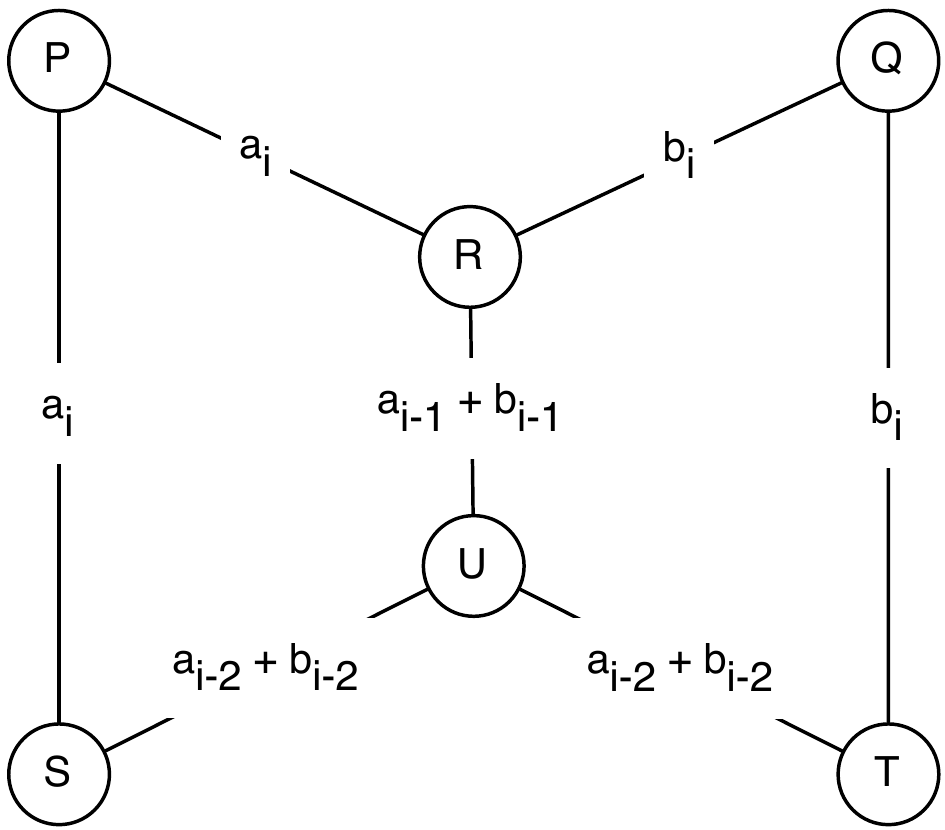}}
   \caption{Butterfly network $N_2$.}
   \label{butterfly_graph}
\end{figure}

Network coding studies methods of attaining maximum information flow in a network where channels have limited throughput. 
A standard example of network coding is given in terms of the butterfly network \cite{acly00} depicted in Figure~\ref{butterfly_graph} as $N_2$. Suppose that parties $P$ and $Q$ generate streams of 1-bit messages $a_1,a_2,\dots$ and $b_1,b_2,\dots$, respectively, with rate one message per second. They need to transmit both sequences of messages to both $S$ and $T$ using only the available communication channels. Each channel's throughput is one bit per second. Note that any protocol over $N_2$ that attempts to independently transmit streams of messages $\{a_i\}_i$ and $\{b_i\}_i$ will fail due to the limited combined capacity of the three channels connecting parties $P$, $Q$, and $R$, with parties $S$, $T$, and $U$.  


The desired result, however, can be easily achieved by a ``network coding" protocol that combines the two streams. 
Under this protocol, at time $1$, party $P$ transmits bit $a_1$ to both $S$ and $R$. At the same time, party $Q$ transmits bit $b_1$ to both $T$ and $R$. At time 2, party $R$ already possesses bits $a_{1}$ and $b_{1}$, so can compute the bit $a_{1} + b_{1} \bmod 2$ and send it to $U$. At time 3, party $U$ forwards this bit to $S$ and $T$. Note that party $S$ received bit $a_1$ directly from party $P$, and after receiving $a_1+b_1 \bmod 2$ from $U$ one second later, $S$ can reconstruct the value of $b_1$, since 
$$a_1+ (a_1+b_1)  \equiv b_1 \pmod 2.$$ Similarly, party $T$ receives $b_1$ directly from $Q$, and can reconstruct the Boolean value $a_1$ after receiving the sum from $U$.  For each time $i>1$, the propagation of bits $a_i$ and $b_i$ is carried out in a similar fashion.

The coding protocol described above can be viewed as a protocol over a collaboration network if the whole stream of messages sent over a single channel in the coding network is interpreted as a single message in the collaboration network.  The computation rules
of the coding protocol are viewed as the local conditions of the collaboration network. For example, if the notation $m_{XY}$ denotes the entire secret value shared between parties $X$ and $Y$, and $[m_{XY}]_i$ denotes its $i$-th bit, then, for example, the local condition at party $R$ can be described as 
$$ \forall i\ge 1\; \left([m_{RU}]_{i+1}\equiv [m_{PR}]_i + [m_{QR}]_i \pmod 2\right).$$ 
The desired properties of the protocol can be stated in our notation as 
$$m_{PS},m_{US}\rhd m_{QR},m_{PR}$$ 
and
$$m_{QT},m_{UT}\rhd m_{QR},m_{PR}.$$ 
Other network protocols that deal with data streams, such as, for example, the alternating bit protocol~\cite{bsw69}, can similarly be interpreted in terms of collaboration networks.

\subsection{Network Topology}

The independence and functional dependence examples we have given so far are for a single protocol, subject to a particular set of local interdependencies between secrets. If the topology remains fixed, but the protocol is changed, then secrets which were previously functionally dependent may no longer be so, and vice versa. For example, for network $N_1$ above, the claim $a\rhd d$ will no longer be true if, say, party $R$ switches from enforcing the local condition $d=b+c \mod 2$ to enforcing the local condition $d=b$. In this paper, we study properties of relations between secrets that follow from the topological structure of the network of channels, no matter which specific protocol is used, as long as it is specified in terms of interdependencies between adjacent channels. Examples of such properties for network $N_1$ are $(a\rhd d) \rightarrow (b,c\rhd d)$ and $(a\parallel b,c) \rightarrow (a\parallel d)$. 

In an earlier work~\cite{mn09a}, we gave a complete axiomatic system for the independence relation between single secrets over a collaboration network. In fact, we axiomatized a slightly more general relation $a_1\parallel a_2\parallel \dots \parallel a_n$ between multiple single secrets.  One can also consider collaboration networks in which a secret is known to any arbitrary subset of parties, rather than a pair of parties.  In a recent paper~\cite{mn10clima}, we generalized the earlier independence results~\cite{mn09a} to this ``hypergraph" setting.
 
In this article, we turn our attention to functional dependence in (non-hypergraph) collaboration networks.  Here, we present a sound and complete logical system that describes the properties of the functional dependence relation $A \rhd B$ between sets of secrets over any fixed network topology $N$.  This system includes Armstrong's {\em Reflexivity}, {\em Augmentation}, and {\em Transitivity} axioms.  To these, we add a {\em \gateway} axiom. The above-mentioned statement $(a\rhd d) \rightarrow (b,c\rhd d)$ is an instance of this new axiom for network $N_1$.  We prove additional statements about different collaboration networks in Section~\ref{sec:example}. 

From the point of view of verification of a specific protocol, the logical calculus introduced in this paper allows us to separate arguments about properties of the protocol itself from the topological properties of the underlying network. For example, since $(a\rhd d) \rightarrow (b,c\rhd d)$ is a property of network $N_1$, if the designers of a particular cryptographic protocol over $N_1$ can guarantee that the value of $d$ can not be reconstructed from the values of $b$ and $c$, then using the axioms of our logical system, one can prove that the value of $d$ is not revealed by the value of $a$ for the same protocol.

\section{Formal Setting}

Throughout this paper, we assume a fixed  infinite alphabet of variables $a, b,\dots$, which we refer to as ``secret variables". By a network topology, we mean a finite graph whose edges, or ``channels", are labeled by secret variables.  We allow loop edges and multiple edges between the same pair of parties. The set of all channels of network $N$ will be denoted by $Ch(N)$.  One channel may have (finitely) many labels, but the same label can be assigned to only one channel. Given this, we will informally refer to ``the channel labeled with $a$" as simply ``channel $a$". 

\begin{definition}\label{}
A semi-protocol over a network $N$ is a pair $\langle  V, L\rangle$ such that
\begin{enumerate}
\item $V(c)$ is an arbitrary set of ``values" for each channel $c \in Ch(N)$,
\item $L = \{L_p\}_{p \in P}$ is a family of predicates, indexed by set $P$ of all parties of the network $N$, which we call ``local conditions".  If $c_1,\dots c_k$ is the list of all channels incident with party $p$, then  $L_p$ is a predicate on $V(c_1)\times\dots\times V(c_k)$.
\end{enumerate}
\end{definition}

\begin{definition}\label{}
A run of a semi-protocol $\langle V, L\rangle$ is a function $r$ such that 
\begin{enumerate}
\item $r(c)\in V(c)$ for any channel $c \in Ch(N)$, 
\item If $c_1,\dots c_k$ is the list of all channels incident with a party $p\in P$, then predicate $L_p(r(c_1),\dots,r(c_k))$ is true.
\end{enumerate}
\end{definition}

\begin{definition}\label{protocol}
A protocol is any semi-protocol that has at least one run.
\end{definition}
The set of all runs of a protocol $\pp$ is denoted by $\rr(\pp)$.

\begin{definition}\label{rank}
A protocol $\pp =\langle V, L\rangle$ is called finite if the set $V(c)$ is finite for every $c\in Ch(N)$. 
\end{definition}

We conclude this section with the key definition of this paper. It is the definition of functional
dependence between sets of channels. 

\begin{definition}\label{dependence}
A set of channels $A=\{a_1,\dots,a_n\}$ functionally determines a set of channels $B=\{b_1,\dots,b_k\}$, with respect to a fixed protocol $\pp$, if
$$
\forall r,r'\in \rr(\pp)\; \left(\bigwedge_{i\le n} r(a_i)=r'(a_i) \rightarrow \bigwedge_{j\le k} r(b_j)=r'(b_j)\right).
$$
\end{definition}

We find it convenient to use the notation $f\equiv_X g$ if functions $f$ and $g$ are equal on every argument from set $X$. 
Using this notation, we can say that a set of channels $A$ functionally determines a set of channels $B$ if 
$$\forall r,r'\in \rr(\pp)\;(r\equiv_A r' \rightarrow r\equiv_B r').$$

\section{Language of Secrets}

By $\Phi(N)$, we denote the set of all properties of secrets in collaboration network $N$ definable through the predicate $A\rhd B$.  More formally, $\Phi(N)$ is a minimal set of formulas defined recursively as follows: (i)  for any two finite sets of secret variables (labels of channels in network $N$) $A$ and $B$, formula $A\rhd B$ is in $\Phi(N)$, (ii) the false constant $\bot$ is in set $\Phi(N)$, and (iii) for any formulas $\phi$ and $\psi \in \Phi(N)$, the implication $\phi\rightarrow \psi$ is in $\Phi(N)$. As usual, we assume that conjunction, disjunction, and negation are defined through $\rightarrow$ and $\bot$. 

Next, we define a relation $\vDash$ between a protocol and a formula from $\Phi(N)$.  Informally, $\pp \vDash\phi$ means that formula $\phi$ is true under protocol $\pp$. 
\begin{definition}\label{}
For any protocol $\pp$ over a network $N$, and any formula $\phi\in\Phi(N)$, we define the relation $\pp \vDash\phi$ recursively as follows:
\begin{enumerate}
\item $\pp \nvDash \bot$,
\item $\pp \vDash A\rhd B$ if the set of channels $A$ functionally determines set of channels $B$ under protocol $\pp$,
\item $\pp \vDash \phi_1\rightarrow\phi_2$ if $\pp \nvDash \phi_1$ or $\pp \vDash \phi_2$.
\end{enumerate}
\end{definition}
In this paper, we study the formulas $\phi \in \Phi(N)$ that are true under {\em any} protocol $\pp$ over fixed network $N$. 
Below we describe a formal logical system for such formulas.  This system, like earlier systems defined by Armstrong~\cite{a74}, More and Naumov~\cite{mn09,mn09a,mn10clima} and by Kelvey, More, Naumov, and Sapp~\cite{kmns10},  belongs to the set of deductive systems that capture properties of secrets.  In general, we refer to such systems as {\em logics of secrets}. Since this paper is focused on only one such system, here we call it {\em the Logic of Secrets}. Before stating the axioms of the Logic of Secrets, we need one more technical definition.

By a path in a network, we mean any undirected path in the graph formed by the channels of the network.
We say that a set of channels $G$ is a {\em gateway} between sets of channels $A$ and $B$ if any path from $A$ to $B$ goes through $G$.  We state this more formally below:

\begin{definition}\label{gateway}
Let $A$, $B$, and $G$ be any three sets of channels in $Ch(N)$. Set $G$ is a gateway between sets $A$ and $B$ if
for any path $(c_1,\dots,c_n)$ in network $N$,
\begin{equation}\label{path cond}
c_1\in A \wedge c_n\in B \rightarrow \bigvee_{1\le i\le n} c_i\in G.
\end{equation}
\end{definition}
Note that in the above definition sets $A$, $B$, and $G$ are not necessarily disjoint. Thus, for example, for any set $A\subseteq Ch(N)$,
set $A$ is a gateway between $A$ and itself. Also, note that the empty set is a gateway between any two components of the network that are not connected to one another.

\section{Axioms}

For a fixed collaboration network $N$, the Logic of Secrets, in addition to propositional tautologies and the Modus Ponens inference rule, contains the following axioms:

\begin{enumerate}
\item {\em Reflexivity}: $A\rhd B$, if $A\supseteq B$,
\item {\em Augmentation}: $A\rhd B \rightarrow A,C\rhd B,C$,
\item {\em Transitivity}: $A\rhd B \rightarrow (B\rhd C \rightarrow A\rhd C)$,
\item {\em \gateway}: $A\rhd B\rightarrow G\rhd B$, if $G$ is a gateway between sets $A$ and $B$ in network $N$.
\end{enumerate}
Recall that the first three of these axioms were introduced by Armstong~\cite{a74}, and they are known in database theory as Armstrong's axioms \cite[p.~81]{guw09}. The soundness of all four axioms will be shown in Section~\ref{sec:sound}.

We use the notation $X \vdash_N \Phi$ to state that formula $\Phi$ is derivable from the set of formulas $X$ in the Logic of Secrets for network $N$.

\section{Examples of Proofs}\label{sec:example}
We will give three examples of proofs in the Logic of Secrets. Our first example refers to square collaboration network $N_3$ depicted in Figure~\ref{square_graph}.

\begin{figure}[htbp]
   \centering
	\scalebox{.5}{\includegraphics{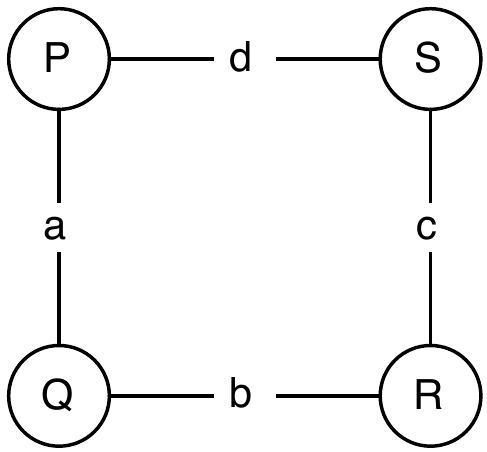}}
   \caption{Network $N_3$.}
   \label{square_graph}
\end{figure}


\begin{proposition}\label{square example theorem}
$\vdash_{N_3} (a\rhd c) \wedge (b\rhd d) \rightarrow (a\rhd d) \wedge (b\rhd c)$.
\end{proposition}

\begin{proof}
Due to the symmetry of the network, it is sufficient to show that $(a\rhd c) \wedge (b\rhd d) \rightarrow a\rhd d$.
Note that $\{a,c\}$ is a gateway between sets $\{b\}$ and $\{d\}$. Thus, by the \gateway axiom,
$b\rhd d$ implies $(a,c\rhd d)$. On the other hand, by the Augmentation axiom, the assumption $a\rhd c$ yields $(a\rhd a, c)$. 
By the Transitivity axiom,  $(a \rhd a,c)$ and $(a,c \rhd d)$ imply $a\rhd d$. \qed
\end{proof}
For the second example, consider the linear network $N_4$ shown in Figure~\ref{linear_graph_5edges}.
\begin{figure}[htbp]
   \centering
	\scalebox{.5}{\includegraphics{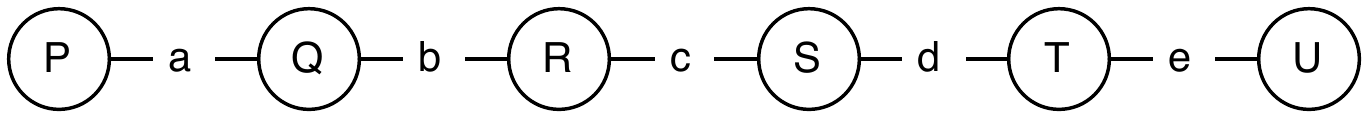}}
   \caption{Network $N_4$.}
   \label{linear_graph_5edges}
\end{figure}

\begin{proposition}\label{linear example theorem}
$\vdash_{N_4} (a\rhd d) \wedge (e\rhd c) \rightarrow b\rhd c$.
\end{proposition}

\begin{proof}
We begin with the assumption that $e\rhd c$.  Since $\{d\}$ is a gateway between sets $\{e\}$ and $\{c\}$, by the \gateway axiom, $d\rhd c$.  Next, using the assumption that $a \rhd d$, the Transitivity axiom yields $a \rhd c$.  Finally, we note that $\{b\}$ is a gateway between $\{a\}$ and $\{c\}$, and apply the \gateway axiom once again to conclude that $b \rhd c$.
\qed \end{proof}

Note that the second hypothesis in the example above is significant.  Indeed, imagine a protocol on $N_4$ where $V(d) = \{0\}$, the set of values allowed on all other channels is $\{0,1\}$, and the local condition at each party $p$ is simply $L_p \equiv true$.  In this protocol, $a \rhd d$ since the value of $a$ on any run clearly determines the (constant) value of $d$.   However, the value of $b$ is of no help in determining the value of $c$, so the conclusion $b \rhd c$ does not hold.

\begin{figure}[htbp]
\begin{center}
\scalebox{.5}{\includegraphics{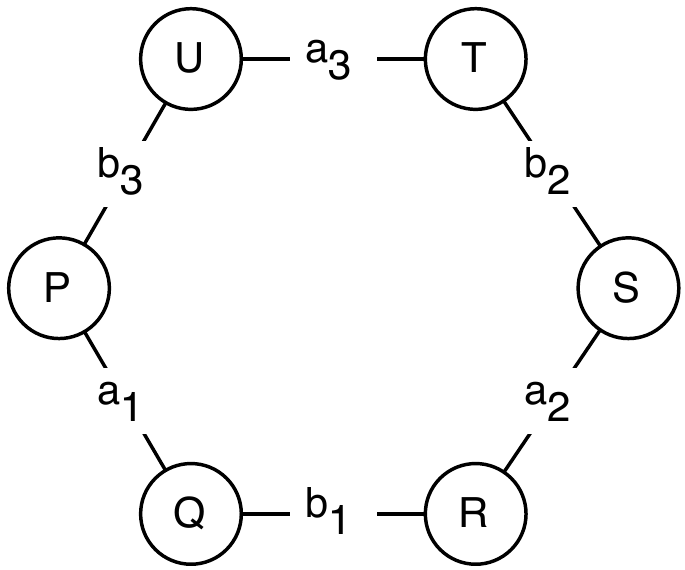}}
\caption{Network $N_4$.}
\label{benzene_graph}
\end{center}
\end{figure}
As our final example, we prove a property of hexagonal collaboration network $N_5$ shown in Figure~\ref{benzene_graph}.

\begin{proposition}\label{}
$\vdash_{N_5}(a_1,a_2\rhd a_3)\wedge(a_2,a_3\rhd a_1)\wedge(a_3,a_1\rhd a_2)\rightarrow b_1,b_2,b_3\rhd a_1,a_2,a_3$.
\end{proposition}
\begin{proof}
Note that $\{b_1,b_3\}$ is a gateway between sets $\{a_2,a_3\}$ and $\{a_1\}$. Thus, by the Gateway axiom,
$(a_2,a_3\rhd a_1)\rightarrow (b_1,b_3\rhd a_1)$. Hence, by the assumption, $(a_2,a_3\rhd a_1)$, we have that $(b_1,b_3\rhd a_1)$. Similarly one
can show that $(b_1,b_2\rhd a_2)$ and $(b_2,b_3\rhd a_3)$ using the assumptions $(a_3,a_1\rhd a_2)$ and $(a_1,a_2\rhd a_3)$.

Consider statements $(b_1,b_3\rhd a_1)$ and $(b_1,b_2\rhd a_2)$. By the Augmentation axiom, they, respectively, imply that
$(b_1,b_2,b_3\rhd a_1,b_1,b_2)$ and $(a_1,b_1,b_2\rhd a_1,a_2)$. Thus, by the Transitivity axiom, $(b_1,b_2,b_3\rhd a_1,a_2)$.

Now consider $(b_1,b_2,b_3\rhd a_1,a_2)$ and statement $(b_2,b_3\rhd a_3)$, established earlier.
By the Augmentation axiom, they, respectively, imply that
$(b_1,b_2,b_3\rhd a_1,a_2,b_2,b_3)$ and $(a_1,a_2,b_2,b_3\rhd a_1,a_2,a_3)$. Thus, by the Transitivity axiom, 
$(b_1,b_2,b_3\rhd a_1,a_2,a_3)$.
\qed
\end{proof}

\section{Soundness}\label{sec:sound}
In this section, we demonstrate the soundness of each of the four axioms in the Logic of Secrets.

\begin{theorem}[Reflexivity]\label{}
$\pp \vDash A\rhd B$, for any protocol $\pp$ and any $B\subseteq A$.
\end{theorem}
\begin{proof}
Consider any two runs $r,r'\in \rr(\pp)$ such that $r\equiv_A r'$. Thus $r\equiv_B r'$ for any $B\subseteq A$.\qed
\end{proof}

\begin{theorem}[Augmentation]\label{}
$\pp \vDash A\rhd B \rightarrow A,C\rhd B,C$, for any protocol $\pp$ and any sets of channels $A$, $B$, and $C$.
\end{theorem}
\begin{proof}
Assume $\pp \vDash A\rhd B$ and consider any two runs $r,r'\in \rr(\pp)$ such that $r\equiv_{A,C}r'$.
By our assumption, $r\equiv_{B}r'$. Therefore, $r\equiv_{B,C}r'$. \qed
\end{proof}

\begin{theorem}[Transitivity]\label{}
$\pp \vDash A\rhd B \rightarrow (B\rhd C \rightarrow A\rhd C)$, for any protocol $\pp$ and any sets of channels $A$, $B$, and $C$.
\end{theorem}
\begin{proof}
Assume $\pp \vDash A\rhd B$ and $\pp \vDash B\rhd C$. Consider any two runs $r,r'\in \rr(\pp)$ such that $r\equiv_{A}r'$.
By the first assumption, $r\equiv_{B}r'$. By the second assumption, $r\equiv_{C}r'$. \qed
\end{proof}

\begin{theorem}[Gateway]\label{}
$\pp \vDash A\rhd B \rightarrow G \rhd B$, for any protocol $\pp$ and any gateway $G$ between sets $A$ and $B$.
\end{theorem}
\begin{proof}
Assume $\pp \vDash A\rhd B$ and consider any two runs $r_1,r_2\in \rr(\pp)$ such that $r_1\equiv_{G}r_2$. We will show that $r_1\equiv_{B}r_2$.
Consider the network $N'$ obtained by removing from $N$ all channels in set $G$. By the definition of a gateway, no single connected component of network $N'$ can contain channels from set $A\setminus G$ and set $B\setminus G$ at the same time. Let us divide all connected components of $N'$ into two subgraphs $N'_A$ and $N'_B$ such that $N'_A$ contains no channels from $B\setminus G$ and $N'_B$ contains no channels from $A\setminus G$. Components that do not contain channels from either $A\setminus G$ or $B\setminus G$ can be arbitrarily assigned to either $N'_A$ or $N'_B$.

Next, define a function $r$ on each $c \in Ch(N)$ as follows:
$$r(c)=\left\{\begin{array}{ll}
                            r_1(c)     & \mbox{ if $c\in N'_A$},\\
                            r_1(c)=r_2(c)\;\;\;\; & \mbox{ if $c\in G$},\\
                            r_2(c)     & \mbox{ if $c\in N'_B$}.
                            \end{array}
                     \right. $$
We will prove that $r$ is a run of protocol $\pp$. We need to show that $r$ satisfies the local conditions of protocol $\pp$ at each party $p$. The connected component of $N'$ containing a party $p$ either belongs to $N'_A$ or $N'_B$. Without loss of generality, assume that it belongs to $N'_A$. Thus, $Inc(p)$, the set of all channels in $N$ incident with party $p$, is a subset of $Ch(N'_A)\cup G$. Hence, $r\equiv_{Inc(p)}r_1$. Therefore, $r$ satisfies the local condition at party $p$ simply because $r_1$ does.

By the definition of $r$, we have $r\equiv_A r_1$ and $r\equiv_B r_2$. Together, the first of these statements and the assumption that $\pp\vDash A\rhd B$ imply that $r\equiv_B r_1$. Thus, due to the second statement, $r_1\equiv_B r \equiv_B r_2$. \qed

\end{proof}

\section{Completeness}
In this section, we demonstrate that the Logic of Secrets is complete with respect to the semantics defined above.  To do so, we first describe the construction of a protocol called $\pp_0$, which is implicitly parameterized by a collaboration network $N$ and a set $X$ of formulas in $\Phi(N)$.

\subsection{Protocol $\pp_0$}
Throughout this section, we will assume that $N$ is a fixed collaboration network, and $X \subseteq \Phi(N)$ is a fixed set of formulas.

\begin{definition}\label{control closure}
For any $A\subseteq Ch(N)$, we define $A^*$ to be the set of all channels $c\in Ch(N)$ such that $X\vdash_N A\rhd c$.
\end{definition}

\begin{theorem}\label{AsubA*}
$A\subseteq A^*$, for any $A\subseteq Ch(N)$.
\end{theorem}
\begin{proof}
Let $a\in A$. By the Reflexivity axiom, $\vdash_N A\rhd a$. Hence, $a\in A^*$.
\qed \end{proof}

\begin{theorem}\label{ArhdA*}
$X\vdash_N A\rhd A^*$, for any $A\subseteq Ch(N)$.
\end{theorem}
\begin{proof}
Let $A^*=\{a_1,\dots,a_n\}$. By the definition of $A^*$,  $X\vdash_N A\rhd a_i$, for any $i\le n$. We will
prove, by induction on $k$, that $X\vdash_N (A\rhd a_1,\dots,a_k)$ for any $0\le k\le n$. 

\noindent {\em Base Case}: $X\vdash_N A\rhd \varnothing$ by the Reflexivity axiom.

\noindent {\em Induction Step}: Assume that $X\vdash_N (A\rhd a_1,\dots,a_k)$. By the Augmentation axiom, 
\begin{equation}\label{eq0}
X\vdash_N A, a_{k+1}\rhd a_1,\dots,a_k,a_{k+1}.
\end{equation}
Recall that $X\vdash_N A\rhd a_{k+1}$. Again by the Augmentation axiom, $X\vdash_N (A\rhd A, a_{k+1})$.
Hence, $X\vdash_N (A \rhd a_1,\dots,a_k,a_{k+1})$, by (\ref{eq0}) and the Transitivity axiom.
\qed \end{proof}

We now proceed to define our protocol $\pp_0$.  We will first specify the set of values $V(c)$ for each channel $c \in Ch(N)$.  In this construction, the value of each channel $c$ on a particular run will be a function from the set $2^{Ch(N)}$ into the set $\{0,1\}$. 
Thus, for any $c\in Ch(N)$ and any $E\subseteq Ch(N)$, we have $r(c)(E)\in \{0,1\}$. 
We will find it more convenient, however, to think about $r$ as a two-argument Boolean function, where $r(c,E)\in \{0,1\}$. 

Furthermore, we will not allow the value of a channel on a particular run to be just {\it any} function from the set $2^{Ch(N)}$ into $\{0,1\}$.  Instead, for any channel $c$, we will restrict set $V(c)$ so that, for any run $r$, if $c\in E^*$, then $r(c,E)=0$. 

To complete the description of protocol $\pp_0$, we will specify the local conditions for each party in the network.
At each party $p$, we define the local condition $L_p$ as
$$\forall E\subseteq Ch(N)\; \forall c,d\in (Inc(p)\setminus E^*)\;  \left(r(c, E)=r(d, E)\right).$$
That is, when two channels are incident with a party $p$ and neither channel is in $E^*$, the values of the functions assigned to those channels on argument $E$ must match on any given run.

To show that $\pp_0$ is indeed a protocol, we only need to show that it has at least one run. Indeed, 
the constant function $r(c,E)=0$ trivially satisfies the local condition at every party $p$.

Now that the definition of protocol $\pp_0$ is complete, we make the following two claims about its relationship to the given set of formulas $X$.

\begin{theorem}\label{th1}
If $\pp_0 \vDash A\rhd B$, then $X \vdash_N A\rhd B$.
\end{theorem}
\begin{proof} 
Assume $\pp_0 \vDash A\rhd B$ and consider two specific runs of $\pp_0$.
The first of these two runs will be the constant run $r_1(c,E)=0$. The second run is defined as
\begin{equation}\label{r2}
r_2(c,E)=\left\{\begin{array}{ll}
                            1     & \mbox{if $c\notin A^*$ and $E=A$},\\
                            0     & \mbox{if $c\in A^*$ or $E\neq A$}.
                            \end{array}
                     \right.
\end{equation}
To show that $r_2$ satisfies the local condition at a party $p$, consider any $E\subseteq Ch(N)$ and 
any $c,d\in Inc(p)\setminus E^*$. If $E\neq A$, then $r_2(c,E)=0=r_2(d,E)$. If $E=A$, then, since $c,d\in Inc(p)\setminus E^*$, we have 
$c,d\notin A^*$. Thus, $r_2(c,E)=1=r_2(d,E)$. Therefore, $r_2$ is a run of protocol $\pp_0$.

Notice that by Theorem~\ref{AsubA*}, $A\subseteq A^*$. Thus, by equality~(\ref{r2}), $r_2(a,E)=0$ for any $a\in A$ and any $E\subseteq Ch(N)$. Hence, $r_1(a,E)=0=r_2(a,E)$ for any $a\in A$ and $E\subseteq Ch(N)$. Thus, by the assumption that $\pp_0 \vDash A\rhd B$, we have $r_1(b,E)=r_2(b,E)$ for any $b\in B$ and $E\subseteq Ch(N)$. In particular, $r_1(b,A)=r_2(b,A)$ for any $b\in B$.
Since, by definition, $r_1(b,A)=0$, we get $r_2(b,A)=0$ for any $b\in B$. By the definition of $r_2$, this means that
$B\subseteq A^*$. By the Reflexivity axiom, $\vdash_N A^*\rhd B$. By Theorem~\ref{ArhdA*} and the Transitivity axiom, $X\vdash_N A\rhd B$.
\qed 

\end{proof}

\begin{theorem}\label{th2}
If $X \vdash_N A\rhd B$, then $\pp_0 \vDash A\rhd B$. 
\end{theorem}
\begin{proof}
Assume that $X \vdash_N A\rhd B$, but $\pp_0 \nvDash A\rhd B$. Thus, there are runs $r_1$ and $r_2$ of $\pp_0$ such that $r_1(a,E)=r_2(a,E)$ for any $a\in A$ and any $E\subseteq Ch(N)$, but there is $b_0\in B$ and $E_0\subseteq Ch(N)$
such that 
\begin{equation}\label{runs not equal}
r_1(b_0,E_0) \neq r_2(b_0,E_0).
\end{equation}
First, assume that network $N'$, obtained from $N$ by the removal of all channels in set $E^*_0$, contains a path $\pi$ connecting channel $b_0$ with a channel $a_0\in A$. Thus, this case implicitly assumes that $b_0,a_0\notin E^*_0$. Let functions $f_1$ and $f_2$ on the channels of network $N$ be defined as $f_1(c)=r_1(c,E_0)$ and $f_2(c)=r_2(c,E_0)$. Due to the local conditions of protocol $\pp_0$, all channels along path $\pi$  must have the same value  of function $f_1$. The same is also true about function $f_2$. Therefore,
$
r_1(b_0,E_0) = f_1(b_0)=f_1(a_0)=r_1(a_0,E_0)=r_2(a_0,E_0)=f_2(a_0)=f_2(b_0)=r_2(b_0,E_0).
$
This is a contradiction with statement (\ref{runs not equal}).

Next, suppose that there is no path in $N'$ connecting $b_0$ with a channel in $A$. Thus, set $E^*_0$ is a gateway between sets $A$ and $\{b_0\}$. By the \gateway axiom,
\begin{equation}\label{eq1}
\vdash_N A\rhd b_0 \rightarrow E^*_0 \rhd b_0.
\end{equation}
By the Reflexivity axiom, $\vdash_N B\rhd b_0$. Recall the assumption $X \vdash_N A\rhd B$. Thus, by the Transitivity axiom, $X \vdash_N A\rhd b_0$. Taking into account (\ref{eq1}), $X \vdash_N E_0^*\rhd b_0$. By Theorem~\ref{ArhdA*}, $\vdash E_0\rhd E^*_0$. Hence, again by Transitivity, 
$X \vdash_N E_0\rhd b_0$. Thus, by Definition~\ref{control closure}, $b_0 \in E_0^*$. Hence, by the definition of protocol $\pp_0$, $r(b_0,E_0)$ has value 0 for any run $r$. Therefore, $r_1(b_0,E_0)=0=r_2(b_0,E_0)$. This is a contradiction with statement (\ref{runs not equal}).
\qed 
\end{proof}

\subsection{Main Result}
Now, we are ready to finish the proof of completeness.

\begin{theorem}\label{}
If $\nvdash_N \phi$, then there is a finite protocol ${\cal P}$ such that ${\cal P}\nvDash \phi$. 
\end{theorem}
\begin{proof}
Assume $\nvdash_N \phi$. Let $X$ be a maximal consistent set of formulas such that $\neg\phi\in X$. Consider the finite  protocol $\pp_0$ parameterized by network $N$ and set of formulas $X$. 
We will show that for any formula $\psi$, $X\vdash_N \psi$ if and only if $\pp_0 \vDash \psi$ by induction on the structural complexity
of formula $\psi$. The base case follows from Theorems~\ref{th1} and~\ref{th2}. The induction case follows from the maximality and consistency of set $X$. To finish the proof of the theorem, select $\psi$ to be $\neg\phi$.
\qed \end{proof}

\begin{corollary}\label{}
Binary relation $\vdash_N\phi$ is decidable.
\end{corollary}

\begin{proof}
This statement follows from the completeness of the Logic of Secrets with respect to {\em finite} protocols and the recursive enumerability of all theorems in the logic. \qed
\end{proof}

\section{Conclusion}

We have presented a complete axiomatization of the properties of the functional dependence relation over secrets on collaboration networks.  In light of previous results capturing properties of the independence relation in the same setting~\cite{mn09a}, it would be interesting to describe properties that connect these two predicates on collaboration networks.  

An example of such a property for the network $N_6$ in Figure~\ref{linear_graph} is given in the following theorem.
\begin{figure}[htbp]
   \centering
	\scalebox{.5}{\includegraphics{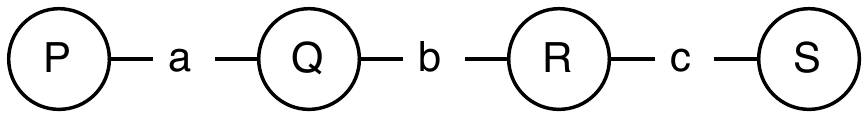}}
   \caption{Network $N_6$.}
   \label{linear_graph}
\end{figure}
\begin{theorem}\label{concl thm}
For any protocol $\pp$ over network $N_6$, 
$$\pp \vDash (a,b\rhd c) \wedge (a \parallel b) \rightarrow b \rhd c.$$
\end{theorem}
\begin{proof}
For any two runs $r_1,r_2 \in \rr(\pp)$ where $r_1(b) = r_2(b)$, we must show that $r_1(c) = r_2(c)$.  The assumption $a \parallel b$ guarantees that values $r_1(a)$ and $r_2(b)$ coexist in some run in $\rr(\pp)$; call this run $r_3$.  Thus, we have $r_3(a)= r_1(a)$ and $r_3(b) =r_2(b)$.  

Next, we create a new function $r_4$ which ``glues" together runs $r_3$ and $r_2$ at party $Q$.  Formally, we define $r_4$ as 
$$r_4(x)=\left\{\begin{array}{ll}
                            r_3(x)     & \mbox{ if $x=a$},\\
                            r_2(x)     & \mbox{ if $x\in \{b,c\}$}.
                            \end{array}
                     \right. $$

We claim that function $r_4$ satisfies the local conditions of protocol $\pp$, since at each party in $N_5$, it behaves locally like an existing run.  Indeed, at party $P$, $r_4$ matches run $r_3$, and at parties $R$ and $S$, $r_4$ matches run $r_2$.  At party $Q$, $r_4$ matches $r_2$ exactly, since $r_4(b)=r_2(b)$.  Thus, $r_4 \in \rr(\pp)$.  To complete the proof, we note that $r_1(a)=r_3(a)=r_4(a)$ and $r_1(b)=r_2(b)=r_4(b)$.  By the assumption that $(a,b\rhd c)$, we have $r_1(c) =r_4(c)$.  The definition of $r_4$ is such that $r_4(c)=r_2(c)$, so $r_1(c)=r_2(c)$, as desired.\qed 
\end{proof}

A complete axiomatization of properties that connect the functional dependence relation and the independence relation between secrets on a collaboration network remains an open problem. 

\section{Acknowledgment}

The authors would like to thank Andrea Mills and Benjamin Sapp for discussions of the functional dependence relation on sets of secrets during earlier stages of this work.

\bibliography{../sp}

\end{document}